%% file: F5.tex
\documentclass[sigconf]{acmart}

\AtBeginDocument{%
  \providecommand\BibTeX{{%
    \normalfont B\kern-0.5em{\scshape i\kern-0.25em b}\kern-0.8em\TeX}}}

\copyrightyear{2020}
\acmYear{2020}
\setcopyright{acmcopyright}\acmConference[ISSAC '20]{International Symposium on Symbolic and Algebraic Computation}{July 20--23, 2020}{Kalamata, Greece}
\acmBooktitle{International Symposium on Symbolic and Algebraic Computation (ISSAC '20), July 20--23, 2020, Kalamata, Greece}
\acmPrice{15.00}
\acmDOI{10.1145/3373207.3404035}
\acmISBN{978-1-4503-7100-1/20/07}

\usepackage[utf8]{inputenc}
\usepackage[T1]{fontenc}

\usepackage{amsmath}
\usepackage{amsthm}
\usepackage{mathrsfs}
\usepackage{stmaryrd}
\usepackage{enumerate}
\usepackage[algoruled,vlined,english,linesnumbered]{algorithm2e}
\providecommand{\DontPrintSemicolon}{\dontprintsemicolon}

\usepackage{comment}
\usepackage{multirow}
\usepackage{xspace}

\usepackage{mathtools}
\mathtoolsset{showonlyrefs,showmanualtags}

\usepackage{ifdraft}

\usepackage[inline]{enumitem}

\usepackage{siunitx}
\usepackage{booktabs}

\SetKwInOut{Input}{input}\SetKwInOut{Output}{output}

\definecolor{input}{HTML}{303060}
\definecolor{output}{HTML}{804000}
\definecolor{string}{HTML}{A02020}
\definecolor{parent}{HTML}{A020A0}
\definecolor{function}{HTML}{205080}
\definecolor{constructor}{HTML}{205080}
\definecolor{method}{HTML}{205080} 
\definecolor{keyword}{HTML}{008000}
\definecolor{error}{HTML}{B01010}
\definecolor{comment}{HTML}{60A060}

\newcommand{\noopsort}[1]{}

\DeclareMathOperator{\val}{val}

\DeclareMathOperator{\lcm}{lcm}

\newcommand{\N}{\mathbb N}
\newcommand{\Z}{\mathbb Z}
\newcommand{\NN}{\mathbb N}
\newcommand{\ZZ}{\mathbb Z}

\newcommand{\Q}{\mathbb Q}

\newcommand{\mon}{\mathrm{mon}}

\newcommand{\X}{\mathbf{X}}

\renewcommand{\i}{\mathbf{i}}
\renewcommand{\j}{\mathbf{j}}
\renewcommand{\r}{\mathbf{r}}

\newcommand{\ifnonempty}[3]{%
  \def\tempa{}%
  \def\tempb{#1}%
  \ifx\tempa\tempb 
  #3 
  \else            
  #2
  \fi}

\newcommand{\Kz}{K^\circ}
\newcommand{\KzX}[1][]{K\{ \X \ifnonempty{#1}{; #1}{} \}^\circ}
\newcommand{\KX}[1][]{K \{ \X \ifnonempty{#1}{; #1}{} \}}
\newcommand{\kX}{k[\X]}

\newcommand{\TzX}[1][]{T\{ \X \ifnonempty{#1}{; #1}{} \}^\circ}

\newcommand{\TT}{\mathbb T}
\newcommand{\TTzX}[1][]{\TT\{ \X \ifnonempty{#1}{; #1}{} \}^\circ}

\newcommand{\Kb}{k}
\newcommand{\KbX}{\Kb[\X]}

\newcommand{\all}{\textrm{all}}

\newcommand{\LT}{LT}

\newcommand{\LM}{LM}
\newcommand{\GB}{\textit{\textsf{GBasis}}}
\newcommand{\regularreduce}{\textsf{regular\_reduce }}

\newcommand{\Jpair}{\textup{\textsf{J-pair}}}

\newcommand{\sage}{\textsc{SageMath}\xspace}

\definecolor{purple}{rgb}{0.6,0,0.6}

\definecolor{answer}{rgb}{0,0.5,0.2}

\newtheorem{theo}{Theorem}[section]
\newtheorem{lem}[theo]{Lemma}
\newtheorem{prop}[theo]{Proposition}

\theoremstyle{definition}

\newtheorem{deftn}[theo]{Definition}

\SetKw{Continue}{continue}
\SetKw{Pop}{pop}
\SetKw{Add}{add}

\begin{document}

\title[Signature-based Algorithms for Gröbner Bases over Tate Algebras]{Signature-based Algorithms\\for Gröbner Bases over Tate Algebras}

\author{Xavier Caruso}
\affiliation{Université de Bordeaux,
  \institution{CNRS, INRIA}
  \city{Bordeaux, France}}
\email{xavier.caruso@normalesup.org}

\author{
  Tristan Vaccon}
\affiliation{Universit\'e de Limoges;
  \institution{CNRS, XLIM UMR 7252}
  \city{Limoges, France}  
  \postcode{87060}  
}
\email{tristan.vaccon@unilim.fr}

\author{
  Thibaut Verron}
\affiliation{
  \institution{Johannes Kepler University\\Institute for Algebra }
  \city{Linz, Austria}  
}
\email{thibaut.verron@jku.at}

\thanks{%
The first author is supported by  the 
ANR grant CLap--CLap,
referenced ANR-18-CE40-0026-01.
The third author is supported by  the 
FWF grant P31571-N32.}

\begin{abstract}
Introduced by Tate in~\cite{Tate}, Tate algebras play 
a major role in the 
context of analytic geometry over the $p$-adics,
where they act as a counterpart to the use
of polynomial algebras in classical algebraic geometry.
In~\cite{CVV} the formalism of Gröbner bases over Tate algebras
has been introduced and effectively implemented.
One of the bottlenecks in the algorithms
was the time spent on reduction, which are
significantly costlier than over polynomials.
In the present article, we introduce two signature-based
Gröbner bases algorithms for Tate algebras, in
order to avoid many reductions.
They have been implemented in \sage. We discuss
their superiority based on numerical evidence.
\end{abstract}



\vspace{-1.5mm}


\maketitle

\hfill
  \emph{This article is dedicated to the memory of John Tate.}

\section{Introduction}

For several decades, many computational questions arising from
geometry and arithmetics have received much attention, leading 
to the development of more and more efficient algorithms and 
software.
A typical example is the development of the theory of Gröbner bases, 
which provides nowadays quite efficient tools for manipulating ideals 
in polynomial algebras and, eventually, algebraic varieties and 
schemes~\cite{Magma, Macaulay, Sage, Singular}.
At the intersection of geometry and number theory, one finds $p$-adic
geometry and, more precisely, the notion of $p$-adic analytic varieties
first defined by Tate in~\cite{Tate} (see also \cite{FP04}), which
plays an
important role in many modern theories and achievements (\emph{e.g.} 
$p$-adic cohomologies~\cite{LS07}, $p$-adic modular forms~\cite{Go88}).

The main algebraic objects upon which Tate's geometry is built are
Tate algebras and their ideals.
In an earlier paper~\cite{CVV}, the authors started to study 
computational aspects related to Tate algebras, introduced
Gröbner bases in this context and designed two algorithms
(adapted from Buchberger's algorithm and the F4 algorithm, respectively)
for computing them.

In the classical setting, the main complexity bottleneck in Gröbner bases 
computations is the time spent reducing
elements modulo the basis.
The most costly reductions are typically reductions to~$0$, because they require
successively eliminating
all terms from the polynomial; yet their output has little value for the rest of the
algorithm. Fortunately, it turns out that
many such reductions can be predicted in advance (for example those
coming from the obvious equality $fg - gf = 0$)
by keeping track of some information on the module representation
of elements of an ideal, called their \emph{signature}.
This idea was first presented in Algorithm F5~\cite{F5} and
led to the development of many algorithms showing different ways to define
signatures, to use them or to compute them.
The interested reader can look at~\cite{EF17} for an extensive survey.

The Tate setting is not an exception to the wisdom that reductions are expensive. The
situation is actually even worse since reductions to $0$ are theorically the result 
of an \emph{infinite} sequence of reduction steps \emph{converging to~$0$}. 
In practice, the process actually stops because we are working at finite precision; however,
the higher the precision is, the more expensive the reductions to~$0$ are, for no benefit.
This observation motivates investigating the possibility of adding signatures to 
Gröbner bases algorithms for Tate series.

\smallskip

\noindent
\textit{Our contribution.}
In this paper, we present two signature-based algorithms for the computation
of Gröbner bases over Tate algebras. They differ in that they use different
orderings on the signatures.

Our first variant, called the PoTe (position over term) algorithm, is 
directly adapted from the G2V algorithm~\cite{GGV10}.
It adopts an incremental point of view
and uses the so-called cover criterion~\cite{GVW16} to detect reductions to~$0$.
A key difficulty in the Tate setting is that the usual way to handle signatures 
assumes the constant term $1$ to be the smallest one. However, this assumption
fails in the Tate setting. We solve this issue by importing ideas 
from the paper~\cite{LWXZ}, in which the case of local algebras is addressed.

In the classical setting, incremental algorithms have the
disadvantage of sometimes computing larger Gröbner bases for intermediate ideals, only to
discard them later on.
In order to mitigate this misfeature,
the F5 algorithm uses a signature ordering taking into account the
degree of the polynomials first, in order to process lower-degree elements first.
In the Tate setting, the degree no longer makes sense and a better measure of 
progression of the algorithms is the valuation.
Nonetheless, in analogy with the classical setting, an incremental
algorithm could perform intermediate computations to high valuation and just discard
them later on.
The second algorithm we will present, called the VaPoTe (valuation over position 
over term) algorithm, uses an analogous idea to that of F5 to mitigate this problem.

\smallskip

\noindent
\textit{Organization of the article.}
In Section~\ref{sec:ingredients}, we recall the basic definitions and
properties of Tate algebras and Gröbner bases over them, together
with the principles of the G2V algorithm. 
Sections \ref{sec:PoTe} and \ref{sec:VaPoTe} are devoted to the 
PoTe and the VaPoTe
algorithms respectively: they are presented and their correctness
and termination are proved.
Finally, implementation, benchmarks and possible future improvements
are discussed in Section~\ref{sec:implementation}.

\smallskip

\noindent
\textit{Notations.}
Throughout this article, we fix a positive integer $n$ and use
the short notation $\X$ for $(X_1, \ldots, X_n)$.
Given $\i = (i_1, \ldots, i_n) \in \NN^n$, we shall write 
$\X^\i$ for $X_1^{i_1} \cdots X_n^{i_n}$.

\section{Ingredients}
\label{sec:ingredients}

In this section, we present the two main ingredients we are going to mix 
together later on. They are, first, the G2V~\cite{GGV10} and 
GVW~\cite{GVW16} signature-based algorithms, and, second, the Tate 
algebras and the theory of Gröbner bases over them as developed 
in~\cite{CVV}.

\subsection{The G2V algorithm}
\label{ssec:G2V}

In what follows, we present the G2V algorithm which was designed by Gao, 
Guan and Volny IV in~\cite{GGV10} as an incremental variant of the 
classical F5 algorithm. Our presentation includes the cover 
criterion which was formulated later on in~\cite{GVW16} by
Gao, Volny IV and Wang.
The incremental point of view is needed for the application we will 
discuss in Section~\ref{sec:VaPoTe}. Moreover we believe that it has two 
extra advantages: first, it leads to simplified notations and, more 
importantly, it shows clearly where intermediate inter-reductions are 
possible.

Let $k$ be a field and $\kX$ denote the ring of polynomials over $k$ 
with indeterminates $\X$. 
We endow $\kX$ with a fixed monomial order $\leq_\omega$.
Let $I_0$ be an ideal in $\kX$. Let $G_0$ be a Gröbner basis of $I_0$ 
with respect to $\leq_\omega$.
Let $f \in \kX$. We aim at computing a GB of the ideal $I = I_0 +
\left\langle f \right\rangle.$
Let $M \subset \kX \times \kX$ be the $\kX$-sub-module defined by
the $(u,v)$ such that $uf-v \in I_0$.
The leading monomial $\LM(u)$ of $u$ is the \emph{signature} of $(u,v)$.

\begin{deftn}[Regular reduction]
\label{def:regred}
  Let $p_1=(u_1,v_1)$ and $p_2=(u_2,v_2)$ be in $M$.  We say that $p_1$ is
  \textit{top-reducible} by $p_2$ if
  \begin{enumerate}
    \item either $v_2=0$ and $\LM(u_2)$ divides $\LM(u_1)$,
    \item or $v_1v_2 \neq 0$, $\LM(v_2)$ divides $\LM(v_1)$ and:
      $$\frac{\LM(v_1)}{\LM(v_2)} \cdot \LM (u_2) \leq \LM(u_1).$$
  \end{enumerate} 
  The corresponding top-reduction is
  \[p = p_1-tp_2=(u_1 - t u_2, v_1-t v_2) \] where $t = \frac{\LM(u_1)}{\LM(u_2)}$ is the
  first case and $t = \frac{\LM(v_1)}{\LM(v_2)}$ in the second case.  This top-reduction
  is called \textit{regular} when $\LM(u_1) > t \LM(u_2)$, that is when the signature of the
  reduced pair $p$ agrees with that of $p_1$; it is called \textit{super} otherwise.
\end{deftn}


\begin{deftn}[Strong Gröbner basis]
\label{def:SGB}
  A finite subset $G$ of $M$
  is called a \textit{strong Gröbner basis} (SGB, for short) of $M$ if
  any nonzero $(u,v) \in M$ is top-reducible by some element of $G$.
\end{deftn}

The G2V strategy derives the computation of a Gröbner basis through 
the computation of an SGB. They are related through the following 
proposition.

\begin{prop}
\label{prop:SGBtoGB}
  Suppose that $G= \{(u_1,v_1),\dots,(u_s,v_s)  \}$ is an SGB of $M.$ Then:
  \begin{enumerate}
    \item $\{ u \textrm{ s.t. } (u,0) \in G \}$ is a Gröbner basis of $(I_0{:}f).$
    \item $\{ v \textrm{ s.t. } (u,v) \in G \textrm{ for some } u \}$ is a Gröbner basis of $I.$
  \end{enumerate}
\end{prop}

To compute an SGB, we rely on J-pairs instead of S-polynomials.

\begin{deftn}[J-pair] 
\label{def:Jpair}
  Let  $p_1=(u_1,v_1)$ and $p_2=(u_2,v_2)$ be two elements in $M$
  such that $v_1 v_2 \neq 0$.
  Let $t = \lcm (\LM(v_1), \LM(v_2))$ and set $t_i = t / \LM(v_i)$
  for $i \in \{1,2\}$. Then:

  \noindent $\bullet$\hspace{1ex}%
  if $t_1 \LM(u_1) < t_2 \LM(u_2)$, the \emph{J-pair} of $(p_1,p_2)$
  is $t_2 p_2$,

  \noindent $\bullet$\hspace{1ex}%
  if $t_1 \LM(u_1) > t_2 \LM(u_2)$, the \emph{J-pair} of $(p_1,p_2)$
  is $t_1 p_1$,

  \noindent $\bullet$\hspace{1ex}%
  if $t_1 \LM(u_1) = t_2 \LM(u_2)$, the \emph{J-pair} of $(p_1,p_2)$
  is not defined.
\end{deftn}

%

\begin{deftn}[Cover]
\label{def:cover}
  We say that $p=(u,v)$ is \emph{covered} by $G \subset M$ if
  there is a pair $(u_i,v_i) \in G$ such that
  $\LM(u_i)$ divides $\LM(u)$ and:
  $$\frac{\LM(u)}{\LM(u_i)} \cdot \LM(v_i) < \LM(v).$$
\end{deftn}

\begin{theo}[Cover Theorem]
\label{theo:cover_poly}
  Let $G$ be a finite subset of $M$ such that:
  \begin{itemize}
    \item $G$ contains $(1,f)$;
    \item the set $\{g \in \kX : (0,g) \in G\}$ forms a Gröbner basis of $I_{0}$.
  \end{itemize}
  Then $G$ is an SGB of $M$ iff every J-pair of $G$ is covered by~$G$.
\end{theo}

This theorem leads naturally to the G2V algorithm
(see~\cite[Fig.~1]{GGV10}) which is rephrased hereafter in 
Algorithm~\ref{algo:PoTe} (page~\pageref{algo:PoTe}).
We underline that, in Algorithm~\ref{algo:PoTe}, the SGB does not 
entirely appear. Indeed, we remark that one can always work with pairs 
$(\LM(u), v)$ in place of $(u,v)$, reducing then drastically the memory 
occupation and the complexity.
The algorithm maintains two lists $G$ and $S$ which are
related to the SGB in construction as follows:
$G \cup (S \times \{0\})$ is equal to the set of all $(\LM(u), v)$
when $(u,v)$ runs over the SGB. The criterion coming from the
cover theorem is implemented on lines~\ref{line:crit1} 
and~\ref{line:crit2}: the first (resp. the second) statement 
checks if $(u,v)$ is covered by an element of $G$ (resp. an element
of $S \times \{0\}$).


\medskip

\noindent
\textit{Syzygies.}
The G2V algorithm does not give a direct access to the module of 
syzygies of the ideal.
However, it does give access to a GB of $(I_0{:}f)$ (see Proposition~\ref{prop:SGBtoGB}), from 
which one can recover partial information about the syzygies, as shown 
below.

\begin{deftn}
\label{def:syzygy}
Given $f_1,\dots,f_m \in \kX$, we define
\[Syz(f_1,\dots,f_m) =
\Big\{\,\, (a_1,\dots,a_m) \in \kX^m \,\textrm{ s.t.} \, \sum_{i=1}^m a_i f_i=0 \,\, \Big\}. \]
\end{deftn}

\begin{lem}
\label{lem:syzygy}
Let $f_1,\dots,f_m$ generate $I_0$ and
let $u_1,\dots,u_s$ generate $(I_0{:}f)$.
For $i \in \{1, \ldots, s\}$, we write 
$$-u_i f = a_{i,1}f_1+\dots+a_{i,m}f_m
\qquad
(a_{i,j} \in \kX)$$
and define
$z_i=(a_{i,1},\dots,a_{i,m},u_i) \in Syz(f_1,\dots,f_m,f)$.
Then \[Syz(f_1,\dots,f_m,f)=(Syz(f_1,\dots,f_m) \times \{ 0 \})+\left\langle z_1,\dots,z_s \right\rangle.\]
\end{lem}

\begin{proof}
Let $(a_1,\dots,a_m, u) \in Syz(f_1,\dots,f_m,f)$.
Then $u \in (I_0{:}f)$
and we can write 
$u=\sum_{i=1}^s b_i u_i.$
Then the syzygy $(a_1,\dots,a_m,u)-\sum_{i=1}^s b_i z_i$ has its last coordinate equal to $0$
and thus belongs to $(Syz(f_1,\dots,f_m) \times \{ 0 \})$, which is enough to conclude. 
\end{proof}

\subsection{Tate algebras}
\label{ssec:tate}

\noindent
\textit{Definitions.}
We fix a field $K$ equipped with a discrete valuation $\val : K \to \Z 
\sqcup \{+\infty\}$, normalized by $\val(K^\times) = \Z$. We assume 
that $K$ is complete with respect to the distance defined by $\val$. We 
let $\Kz$ be the subring of $K$ consisting of elements of nonnegative 
valuation and $\pi$ be a uniformizer of $K$, that is an element of 
valuation $1$. We set $\Kb = \Kz/\pi\Kz$.
The Tate algebra $\KX$ is defined by:
\begin{equation}
\label{eq:Tatealg}
\KX := \Big\{ \sum_{\i \in \NN^{n}} a_{\i}\X^{\i}
\text{ s.t. }
a_{\i}\in K \text{ and }
\val(a_\i) \xrightarrow[|\i| \rightarrow +\infty]{} +\infty
\Big\}
\end{equation}
Series in $\KX$ have a natural analytic interpretation: they are 
analytic functions on the closed unit disc in $K^n$. We recall that
$\KX$ is equipped with the so-called Gauss valuation defined by:
$$\val\Big(\sum_{\i \in\N^n} a_\i X^\i\Big) =
\min_{\i\in \N^n} \val(a_\i).$$
Series with nonnegative valuation form a subring $\KzX$ of $\KX$.
The reduction modulo $\pi$ defines a surjective homomorphism of
rings $\KzX \to \KbX$.

\medskip

\noindent
\textit{Terms and monomials.}
By definition, an integral \emph{Tate term} is an expression of 
the form $a \X^\i$ with $a \in \Kz$, $a \neq 0$ and $\i \in \NN^n$.
Integral Tate terms form a monoid, denoted by $\TzX$, which is abstractly
isomorphic to $(\Kz \backslash\{0\}) \times \NN^n$. We say that two Tate
terms $a \X^\i$ and $b \X^\j$ are equivalent when $\val(a) = 
\val(b)$ and $\i = \j$. Tate terms modulo equivalence define a
quotient $\TTzX$ of $\TzX$, which is isomorphic to $\N \times
\NN^n$. The image in $\TTzX$ of a term $t \in \TzX$ is called
the \emph{monomial} of $t$ and is denoted by $\mon(t)$.

We fix a monomial order $\leq_\omega$ on $\N^n$ and order
$\TTzX \simeq \N \times \NN^n$ lexicographically by block with respect 
to the reverse natural ordering on the first factor $\N$ and the order 
$\leq_\omega$
on $\NN^n$. Pulling back this order along the morphism $\mon$, we
obtain a preorder of $\TzX$ that we shall continue to denote by $\leq$.
The \emph{leading term} of a Tate 
series $f = \sum a_\i \X^\i \in \KzX$ is defined by:
$$\LT(f) = \displaystyle \max_{\i \in \N^n} \,a_\i X^\i \in \TzX.$$
We observe that the $a_\i X^\i$'s are pairwise nonequivalent in $\TzX$, 
showing that there is no ambiguity in the definition of $\LT(f)$. The
\emph{leading monomial} of $f$ is by definition $\LM(f) = \mon(\LT(f))$.

\medskip

\noindent
\textit{Gröbner bases.}
The previous inputs allow us to define the notion of Gröbner bases for 
an ideal of $\KzX$.

\begin{deftn}
Let $I$ be an ideal of $\KzX$.
A family $(g_1,\dots,g_s) \in I^s$ is a Gröbner basis (in short, GB)
of $I$ if, for all $f \in I$, there exists $i \in \{1, \ldots, s\}$
such that $\LM(g_i)$ divides $\LM(f)$.
\end{deftn}

A classical argument shows that any GB of an ideal $I$ generates $I$.
The following theorem is proved in \cite[Theorem~2.19]{CVV}.

\begin{theo}
Every ideal of $\KzX$ admits a GB.
\end{theo}

The explicit computation of such a GB is of course a central question.
It was addressed in~\cite{CVV}, in which the authors
describe a Buchberger algorithm and an F4 algorithm for this task.
The aim of the present article is to improve on these results by
introducing signatures in this framework and eventually design
F5-like algorithms for the computation of GB over Tate algebras.

\medskip

\noindent
\textit{Important remarks.}
For the simplicity of exposition, we chose to restrict ourselves to 
the Tate algebra $\KX$ and not consider the variants $\KX[\r]$ allowing 
for more general radii of convergence.
However, using the techniques developed in \cite{CVV} (paragraph
\emph{General log-radii} of Section 3.2), all the results we will obtain 
in this article can be extended to $\KX[\r]$.

In practice, the elements of $K$ need to be truncated to fit in the
memory of the computer; when doing so, we say that we are working at
\emph{finite precision}. We refer to~\cite{CVV} (see in particular
Theorem~3.8 and comments around it) for a thorough study of the
behaviour of GB with respect to finite precision computations.

\section{Position over term}
\label{sec:PoTe}

The goal of this section is to adapt the G2V algorithm to the setting 
of Tate algebras. Although all definitions, statements and algorithms
are \emph{formally} absolutely parallel to the classical 
setting, proofs in the framework of Tate algebras are more subtle,
due to the fact that the orderings on Tate terms are not well-founded but
only topologically well-founded. In order to accomodate this 
weaker property, we import ideas from~\cite{LWXZ} where the case
of local rings is considered.

\subsection{The PoTe algorithm}

  \input{algos}

We fix a monomial order $\leq_\omega$ of $\NN^n$ and write $\leq$
for the term order on $\TzX$ it induces.
We consider an ideal $I_0$ in $\KzX$ along with a GB $G_0$ of $I_0$.
Let $f \in \KzX$. We are interested in computing a GB of 
$I=I_0+\left\langle f \right\rangle$.
Mimicking what we have recalled in \S \ref{ssec:G2V},
we introduce the $\KzX$-sub-module $M \subset \KzX \times \KzX$ 
consisiting of pairs $(u,v)$ such that $uf-v \in I_0$.
The definitions of regular reduction (Definition~\ref{def:regred}),
strong Gröbner bases (Definition~\ref{def:SGB}), J-pair
(Definition~\ref{def:Jpair}) and cover (Definition~\ref{def:cover})
extend \emph{verbatim} to the context of Tate algebras, with
the precaution that the leading monomial is now computed with 
respect to the order $\leq$ as explained in Section \ref{ssec:tate}.

\begin{prop}
  Suppose that $G= \{(u_1,v_1),\dots,(u_s,v_s)  \}$ is an SGB of $M.$ Then:
  \begin{enumerate}
    \item $\{ u \textrm{ s.t. } (u,0) \in G \}$ is a Gröbner basis of $(I_0 : f).$
    \item $\{ v \textrm{ s.t. } (u,v) \in G \textrm{ for some } u \}$ is a Gröbner basis of $I.$
  \end{enumerate}
\end{prop}

\begin{proof}
Let $G$ be an SGB of M.

Let $h \in (I_0{:}f).$ Then $hf \in I_0$ and $(h,0) \in M$.
By definition, since $G$ is an SGB of $M$, there exists $(u,0) \in G$
such that $\LM(u)$ divides $\LM(h)$. This implies the first statement
of the proposition.

Let now $h \in I$. If $\LM(h) \in I_0$, there exists a pair $(0,h')
\in M$ with $\LM(h) = \LM(h')$. This pair is divisible by 
some $(0,v) \in G$, proving that $\LM(v)$ divides $\LM(h') = \LM(h)$ 
in this case. We now suppose that $\LM(h) \not\in \LM(I_0)$.
This assumption implies that any $a \in \KzX$ with $(a,h) \in M$
(\emph{i.e.} $af - h \in I_0$)
must satisfy $\LM(a) \geq \LM(h)/\LM(f)$. We can then choose 
a series $a \in \KzX$ such that $(a, h) \in M$ and $\LM(a)$ is 
minimal for this property. Moreover, since $G$ is an SGB, the pair 
$(a,h)$ has to be top-reducible by some $(u,v) \in G$.
If $v \neq 0$, we deduce that $\LM(v)$ divides $\LM(h)$. Otherwise, 
letting $t = \LT(a)/ \LT(u)$, we obtain $(a-tu,h) \in M$ with 
$\LM(a-tu) < \LM(a)$, contradicting the minimality of $\LM(a)$.
As a conclusion, we have shown that $\LM(v)$ divides $\LM(h)$ in
all cases, which readily implies~(2).
\end{proof}

\begin{theo}[Cover Theorem]
\label{theo:cover}
  Let $G$ be a finite subset of $M$ such that:
  \begin{itemize}
    \item $G$ contains $(1,f)$;
    \item the set $\{g \in \KzX : (0,g) \in G\}$ forms a Gröbner basis of $I_{0}$.
  \end{itemize}
  Then $G$ is an SGB of $M$ iff every J-pair of $G$ is covered by~$G$.
\end{theo}

The proof of Theorem~\ref{theo:cover} is presented in 
Section \ref{ssec:proofcover} below. Before this, let us observe that 
Theorem~\ref{theo:cover} readily shows that the G2V algorithm (see 
Algorithm~\ref{algo:PoTe}) extends 
\emph{verbatim} to Tate algebras. The resulting algorithm is called the 
PoTe\footnote{PoTe means ``\textbf{Po}sition over \textbf{Te}rm''.} 
algorithm.
The correctness of the PoTe algorithm is clear thanks to Theorem 
\ref{theo:cover}. Its termination is not \emph{a priori} guaranteed 
because the call to \regularreduce may enter an infinite loop 
(see \cite[Sec. 3.1]{CVV}). However, if we assume that all regular 
reductions terminate (which is guaranteed in practice by working at 
finite precision), the PoTe algorithm terminates as well thanks to the 
Noetherianity of $\KzX$.

\subsection{Proof of the cover theorem}
\label{ssec:proofcover}

Throughout this subsection, we consider a finite set $G$ satisfying the 
assumptions of Theorem~\ref{theo:cover}.

We first assume that $G$ is an SGB of $M$. Let $p_1, p_2 \in G$ and
write $p_i = (u_i, v_i)$ for $i \in\{1,2\}$.
We set $t = \lcm(\LM(v_1), \LM(v_2)) \in \TTzX$ and $t_i = t/\LM(v_i)$.
If $\LM(t_1 u_1) = \LM(t_2 u_2)$, the $J$-pair of $(p_1, p_2)$ is not
defined and there is nothing to prove. Otherwise, if $i$ (resp. $j$) 
is the index for which $\LM(t_i u_i)$ is maximal (resp.
$\LM(t_j u_j)$ is minimal), the $J$-pair of $(p_1, p_2)$ is
$t_i p_i$, which is regularly top-reducible by $p_j$. Continuing to
apply regular top-reductions by elements of $G$ as long as possible, 
we reach a pair $(u_0, v_0) \in M$
which is no longer regularly top-reducible by any element of $G$ and
for which $\LM(u_0) = \LM(t_i u_i)$ and $\LM(v_0) < \LM(t_i v_i)$.
Since $G$ is an SGB of $M$, $(u_0, v_0)$ must be super top-reducible
by some pair $(u,v) \in G$. By definition of super top-reducibility, 
$\LM(u)$ divides $\LM(u_0) = \LM(t_i u_i)$ and $\LM(v)\cdot\LM(u_0) = 
\LM(v_0) \cdot\LM(u)$. This shows that 
$\LM(v)\cdot\LM(u_i) < \LM(v_i)\cdot
\LM(u)$ and then that $(u,v)$ covers $t_i p_i$.

We now focus on the converse and assume that each $J$-pair of $G$
is covered by $G$. We define:
$$W=\big\{\,\, (u,v) \in M, \textrm{ top-reducible by no pair of } G \,\,\big\}$$
and assume by contradiction that $W$ is not empty.

  \begin{lem}
    \label{lem:cover-lem-no-u-eq0}
    The set $W$ does not contain any pair of the form $(u,v)$ with 
    $u = 0$ or $\LM(v) \in \LM(I_{0})$.
  \end{lem}
  \begin{proof}
    By our assumptions, if $\LM(v) \in \LM(I_{0})$, $v$ is reducible
    by some $g$ with $(0,g) \in G$. In particular, $(u,v)$ is 
    top-reducible by $(0,g)$ and cannot be in $W$.
    If $u = 0$, then $v \in I_0$ and we are reduced to the previous case.
  \end{proof}

  \begin{lem}
    \label{lem:cover-lem-covering-csq}
    Let $p_{0}=(u_{0},v_{0}) \in W$.
    Then there exists a pair $p_{1}= (u_{1},v_{1}) \in G$ such that
    $\LT(u_{1})$ divides $\LT(u_{0})$, say $\LT(u_{0}) = t_{1} \LT(u_{1})$,
    and $t_{1} \LT(v_{1})$ is minimal for this property.

    Furthermore, $t_{1} p_{1}$ is not regularly top-reducible by $G$.
  \end{lem}

  \begin{proof}
    We have already noticed that $u_{0}\neq 0$.
    Since $(1,f) \in G$, there exists a pair in $G$ satisfying the first
    condition. Since $G$ is finite, there exists one that further satisfies 
    the minimality condition.

    We assume by contradiction that $t_{1}p_{1}$ is regularly top-reducible by $G$.
    Consider $p_{2}=(u_{2},v_{2}) \in G$ be a regular reducer of $t_{1}p_{1}$, in particular there
    exists a term $t_{2}$ such that $t_{2} \LT(v_{2}) = t_{1}\LT(v_{1})$, and $t_{2}
    \LT(u_{2}) < t_{1}\LT(u_{1})$.
    The J-pair of $p_{1}$ and $p_{2}$ is then defined and equals
    $\tau \cdot (u_{1},v_{1})$
    with $\tau$ dividing $t_{1}$. Write $t_{1} = \tau t'_{1}$ for some term $t'_1$.
    By hypothesis, this J-pair is covered, so there exists $P = (U,V) \in G$ and a term
    $\theta$ such that $\theta \cdot \LT(U) = \tau \cdot \LT(u_{1})$ and 
    $\theta \cdot \LT(V) < \tau \cdot \LT(v_{1})$.
    As a consequence:
    \begin{align*}
    t'_{1} \theta \cdot \LT(U) & = t_{1} \cdot \LT(u_{1}) = \LT(u_{0})\\
    t'_{1} \theta \cdot \LT(V) & < t \cdot \LT(v_{1}).
    \end{align*}
    So $t'_{1}P$ contradicts the minimality of $p_{1}$.
  \end{proof}

  Let $\nu$ be the minimal valuation of a series $v$ for which $(u,v) 
  \in W$. We make the following additional assumption: $\nu < +\infty$.
  In other words, we assume that $W$ contains 
  at least one element of the form $(u,v)$ with $v \neq 0$. We set:
  $$W_1 = \big\{\,\, (u,v) \in W \text{ s.t. } \val(\LM(v)) = \nu \,\, \big\}.$$

  \begin{lem}
    \label{lem:cover-lem-W1-min-u}
    The set $L = \{\LM(u) : (u,v) \in W_{1}\}$ admits a minimal element.
  \end{lem}
  \begin{proof}
    We assume by contradiction that $L$ does not have a minimal element. 
    Thus, we can construct a sequence
    $(u_k,v_k)_{k \geq 1}$ with values in $W_{1}$ such that $\LM(u_k)$ 
    is strictly decreasing.
    As a consequence, in the Tate topology, $u_k f$ converges to $0$.
    Hence, for $k$ large enough, $\val(u_{k}f) > \nu = \val(v_k)$.
    From $W_{1} \subset M$, we get $v_{k} - u_{k}f \in I_{0}$ and $\LM(v_{k}) =
    \LM(v_{k}-u_{k}f) \in \LM(I_{0})$.
    By Lemma~\ref{lem:cover-lem-no-u-eq0}, this is a contradiction.
  \end{proof}

  Let $W_2$ be the subset of $W_1$ consisting of pairs $(u,v)$ for
  which $\LM(u)$ is minimal.
  Note that by Lemma~\ref{lem:cover-lem-no-u-eq0}, this minimal value 
  is nonzero.

  \begin{lem}
    \label{lem:cover-lem-W2-ltv}
    For any $(u_{1},v_{1}), (u_{2},v_{2}) \in W_{2}$, $\LM(v_{1}) = \LM(v_{2})$.
  \end{lem}
  \begin{proof}
    Let $(u_{1},v_{1})$ and $(u_{2},v_{2})$ in $W_{2}$, and assume that the leading terms
    are not equivalent, that is $\LM(v_{1})
    \neq \LM(v_{2})$.
    Without loss of generality, we can assume that $\LM(v_{1}) > \LM(v_{2})$.
    By construction of $W_{2}$, $\LM(u_{1}) = \LM(u_{2})$, that is 
    $\LT(u_{1}) = a \LT(u_{2})$ for some $a \in K$, $\val(a) = 0$.
    Since $u_{1}$ and $u_{2}$ are nonzero, we can write $u_{1} = \LT(u_{1}) + r_{1}$ and
    $u_{2} = \LT(u_{2}) + r_{2}$. Eliminating the leading terms, we obtain a new element
    $(u',v') = (r_{1}-ar_{2}, v_{1}-av_{2})$.
    By assumption, $\LM(v') = \LM(v_{1})$, and $\LM(u') < \LM(u_1)$.
    Observe that $(u',v')$ cannot be top-reduced by $G$ as otherwise, $(u_1,v_1)$ would
    also be top-reducible by $G$. Hence $(u',v') \in W_1$, contradicting the
    minimality of $\LM(u_1)$.
  \end{proof}

  Let now $p_0=(u_0,v_0) \in W_2.$
  From Lemma~\ref{lem:cover-lem-covering-csq}, there exists $p_{1} = (u_{1},v_{1}) \in G$
  and a term $t$ such that $\LT(t u_{1}) = \LT(u_{0})$ and $t p_{1}$ is not regular
  top-reducible by $G$.
  We define \[p_* = (u_*,v_*)=p_0-tp_1=(u_0,v_0)-t (u_1,v_1). \]
  We remark that $\LM(u_*) < \LM (u_0)$. Moreover $\LM (v_0)  \neq \LM(t v_1)$ since
  otherwise $p_0$ would be top-reducible by $p_1$, contradicting the fact
  that $p_0 \in W$.

  We first examine the case where $\LM (v_0) < \LM(t v_1)$.
  It implies that $\LM(v_*) = \LM(t v_{1}) > \LM(v_0)$. Let us 
  prove first
  that $p_* \not\in W$. We argue by contradiction. From $p_*
  \in W$, we would derive $\val(v_*) \geq \nu = \val(v_0)$
  and then $\val(v_*) = \val(v_0)$ since the inequality in the
  other direction holds by assumption. We conclude by noticing
  that $\LM(u_*) < \LM (u_0)$ contradicts the minimality of $\LM(u_0)$.
  So $p_* \not\in W$, \emph{i.e.} $p_{*}$ is top-reducible by $G$.
  Let $p_2 = (u_2,v_2) \in G$ be top-reducing~$p_*$.
  If $v_2=0$, then $\LM (u_2)$ divides $\LM(u_*)$.
  Besides, the pair
  $\textstyle 
  p_*' = (u_*',\,v_*) = \big(u_*-\frac{\LT(u_*)}{\LT (u_2)} u_2,\, v_*\big)$
  satisfies $\LM(u_*')<\LM (u_*)$ and thus cannot be in $W$ either.
  We iterate this process until we can only find a reductor $q=(U,V) 
  \in G$ with $V\neq 0$.
  Let $t_2 = \LM(v_*)/ \LM(V)$.
  Then $t_2 \LM(V) = \LM(v_*) = \LM(t v_1)$ and
  $t_2 \LM(U) \leq \LM (u_*) < \LM (t u_1)$ if $U \neq 0$.
  Therefore $q$ regularly top-reduces $t p_1$, which contradicts
  Lemma~\ref{lem:cover-lem-covering-csq}.

  Let us now move to the case where $\LM (v_0) > \LM(t v_1)$.
  Then $\LM (v_*)= \LM(v_0)$. Since $\LM(u_*) < \LM(u_0)$, 
  it follows that $p_* \notin W$, \emph{i.e.} $p_*$ is top-reducible by $G$.
  As in the previous case, we construct $q =(U,V) \in G$ with $V \neq 0$,
  and a term $t_2$ with the properties that
  $t_2 \LM(V) = \LM(v_*) = \LM(v_0)$ and
  $t_2 \LM(U) \leq \LM (u_*) < \LM (u_0)$
  if $U \neq 0$.
  Thus $q$ regularly top-reduces $p_0$, which contradicts $p_0 \in W$.

  As a conclusion, in both cases, we have reached a contradiction.
  This ensures that $\nu = +\infty$. In particulier, $W$ contains an
  element $p_0$ of the form $(u_0, 0)$.
  Let $p_{1} = (u_{1},v_{1}) \in G$ be given by
  Lemma~\ref{lem:cover-lem-covering-csq}.
  If $v_{1}=0$, this pair would be a reducer of $(u_0,0) \in W$, which is a contradiction.
  So $v_{1} \neq 0$.
  Set $t = \frac{\LT(u)}{\LT(u_1)}$. Let:
  $$p_* = (u_*,v_*) = (u_0,0) - t(u_{1},v_{1}) = (u_0 - tu_1, -v_1)$$
  Then $\LM(u_*) < \LM(u_0)$ and $\LM(v_*) = t \LM(v_{1})$.
  From $v_{1} \neq 0$, we deduce $p_* \notin W$. So $p_*$ is 
  top-reducible by $p_{2} = (u_{2},v_{2}) \in G$, meaning that there
  exists a term $t_1$ such that
  $t_{1} \LM(v_{2}) = \LM(v_*) = t \LM(v_{1})$ and $t_{1}\LM(u_{2}) \leq
  \LM(u_*) < t \LM(u_{1})$.
  So $p_2$ is a regular top-reducer of $t p_{1}$, which contradicts
  Lemma~\ref{lem:cover-lem-covering-csq}.

  Finally, we conclude that $W$ is empty.
  By construction, $G$ is an SGB of $M$.

\section{Valuation over position over term}
\label{sec:VaPoTe}

In this section, we design a variant of the PoTe algorithm in which, 
roughly speaking, signatures are first ordered by increasing valuations.

\subsection{The VaPoTe algorithm}
\label{sec:vapote-algorithm}

The VaPoTe\footnote{VaPoTe means
``\textbf{Va}luation  over \textbf{Po}sition over \textbf{Te}rm''}
algorithm is Algorithm~\ref{algo:VaPoTe} (page~\pageref{algo:VaPoTe}).
It is striking to observe that it looks formally very similar
to the PoTe Algorithm (Algorithm~\ref{algo:PoTe}) as they only differ 
on lines~\ref{line:blank}--\ref{line:popf} and, more importantly, on 
lines~\ref{line:diff1}--\ref{line:diff2}. However, these slight 
changes may have significant consequences on the order in which the 
inputs are processed, implying possibly important differences in the 
behaviour of the algorithms.

The VaPoTe algorithm has a couple of interesting features. First, if 
we stop the execution of the algorithm at the moment when we first reach 
a series $f$ of valuation greater than $N$ on line 4, the value of $\GB$ 
is a GB of the image of $I = \left< f_1, \ldots, f_m \right>$ in 
$\KzX/\pi^N\KzX$. In other words, the VaPoTe algorithm can be used to 
compute GB of ideals of $\KzX/(\pi^N) \simeq \Kz[\X]/(\pi^N)$ (for 
our modified order) as well.

Secondly, Algorithm~\ref{algo:VaPoTe} remains correct if the reduction 
on line~\ref{line:regred} is interrupted as soon as the valuation 
rises. The property allows for delaying some reductions, which might be 
expensive at one time but cheaper later (because more reductors are
available).
It also has a theoretical interest because the reduction process may 
\emph{a priori} hang forever (if we are working at infinite precision);
interrupting it prematurely removes this defect and leads to more
satisfying termination results.

\subsection{Proof of correctness and termination}

We introduce some notation.
For a series $f \in \KzX$, we write $\nu(f) = \pi^{-\val(f)} f$
(which has valuation $0$ by construction) and define $\rho(f)$ as 
the image of $\nu(f)$ in $\KzX/\pi\KzX \simeq \kX$. More generally
if $A$ is a subset of $\KzX$, we define $\nu(A)$ and $\rho(A)$
accordingly.

We consider $f_1, \ldots, f_m \in \KzX$ and
write $I$ for the ideal of $\KzX$ they generate.
For an integer $N$, we set $I_N = I \cap (\pi^N \KzX)$.
Clearly $I_{N+1} \subset I_N$ for all $N$. Let $\bar I_N$ be the 
image of $\pi^{-N} I_N$ in $\kX$; we have a canonical isomorphism
$\bar I_N \simeq I_N/I_{N+1}$. 
Besides, the morphism $I_N \to I_{N+1}$, $f \mapsto \pi f$ induces an 
inclusion $\bar I_N \hookrightarrow \bar I_{N+1}$. Hence, the $\bar I_N$'s form
a nondecreasing sequence of ideals of $\kX$.

We define $Q_\all$ as the set of all series that are popped from 
$Q$ on line~\ref{line:diff1} during the execution of 
Algorithm~\ref{algo:VaPoTe}. Since the algorithm terminates when $Q$ is 
empty, $Q_\all$ is also the set of all series that have been in $Q$ at 
some moment. For an integer $N$, we define
\begin{align*}
Q_{> N} & = \big\{ \, f \in Q_\all \text{ s.t. } \val(f) > N \, \big\}.
\end{align*}
and similarly $Q_{N}$ and $Q_{\leq N}$.
Let also $\tau_N$ be the first time we enter in the while loop on 
line~\ref{line:while} with $Q \subset \pi^N \KzX$. If this event 
never occurs, $\tau_N$ is defined as the time the algorithm exits
the main while loop. We finally let $\GB_N$ be the value of the
variable $\GB$ at the checkpoint $\tau_N$.

\begin{lem}
\label{lem:checkpoints}
Between the checkpoints $\tau_N$ and $\tau_{N+1}$:

\noindent
(1)~the elements popped from $Q$ are exactly those of $Q_N$, and

\noindent
(2)~the ``reduction modulo $\pi^{N+1}$'' of the VaPoTe algorithm 
behaves like the G2V algorithm, with 
input polynomials $\rho(Q_N)$ and initial value of $\GB$ set to
$\rho(\GB_N)$.

\end{lem}

\begin{proof}
We observe that, after the time $\tau_N$, only elements with
valuation at least $N{+}1$ are added to $Q$. The first statement
then follows from the fact that the elements of $Q$ have been 
popped by increasing valuation.
The second statement is a consequence of~(1) together with the fact 
that all $f$
and $v$ manipulated by Algorithm~\ref{algo:VaPoTe} between the
times $\tau_N$ and $\tau_{N+1}$ have valuation $N$.
\end{proof}

Since the G2V algorithm terminates for polynomials over a field, 
Lemma~\ref{lem:checkpoints} ensures that each checkpoint $\tau_N$ is 
reached in finite time if the call to \regularreduce does not 
hang forever. This latter property holds when we are working at finite 
precision and is also guaranteed if we interrupt the reduction as soon 
as the valuation raises.

We are now going to relate the ideals $\bar I_N$ with the sets $Q_N$, 
$Q_{\leq N}$ and $Q_{> N}$. For this, we introduce the syzygies between 
the elements of $\rho(Q_{\leq N})$. More precisely, we set:
$$S_N = \Big\{ \,\,
(a_f)_{f \in Q_{\leq N}} \, \text{ s.t. }
\sum_{f \in Q_{\leq N}} \hspace{-1ex} a_f \nu(f) \equiv 0 \pmod{\pi} 
\,\, \Big\}.$$
and let $\bar S_N$ be the image of $S_N$ under the projection
$\KzX \,{\to}\, \kX$; in other words, $\bar S_N$ is the module of syzygies
of the set $\rho(Q_{\leq N})$, \emph{i.e.} $\bar S_n = Syz(\rho(Q_{\leq N}))$
with the notation of Definition~\ref{def:syzygy}. 
We also define a linear mapping
$\varphi_{N}: (\KzX)^{Q_{\leq N}} \to  \KzX$ by
\begin{displaymath}
  \varphi_{N} : (a_f)_{f \in Q_{\leq N}} \mapsto 
   \sum_{f \in Q_{\leq N}} \hspace{-1ex} a_f \nu(f).
\end{displaymath}
By definition, $\varphi_N$ takes its values in the ideal generated by
$\nu(Q_{\leq N})$ and $\varphi_N(S_N) \subset \pi \KzX$.

\begin{prop}
\label{prop:VaPoTe}
For any integer $N$, the following holds:

\smallskip

(a)~The family $\rho(\GB_{N+1})$ is a GB of $\bar I_N$.

\smallskip

(b)~$\varphi_N(S_N) \subset 
\big< \pi {\cdot} \nu(Q_{\leq N}), \, \pi^{-N} Q_{> N}\big>$.

\smallskip

(c)~$I_{N+1} = \big< \pi^{N+1} {\cdot} \nu(Q_{\leq N+1}), \, Q_{> N+1}\big>$.

\smallskip

(d)~$\bar I_{N+1} = \big< \rho(Q_{\leq N+1}) \big>$.

\end{prop}

\begin{proof}
When $N < 0$, we have $S_N = 0$ and $I_{N+1} = I$,
so that the proposition is obvious.
We now consider a nonnegative integer $N$ and assume that the 
proposition holds for $N{-}1$.
By the induction hypothesis, we know that $\rho(\GB_N)$ is a GB of 
$\bar I_{N-1}$.
It then follows from Lemma~\ref{lem:checkpoints} that 
$\rho(\GB_{N+1})$ is a GB of the ideal generated by $\bar I_{N-1}$ and 
$\rho(Q_N)$, which is equal to $\bar I_N$ by the induction hypothesis. 
The assertion~(a) is then proved.

Between the checkpoints $\tau_N$ and $\tau_{N+1}$,
each signature $u$ added to $S$ on line~\ref{line:diff2} 
corresponds to a family 
$(a_f)_{f \in Q_{\leq N}}$ for which the sum $\sum_f a_f f$ equals
the element $v_0$ added to $Q$ on the same line. Rescaling the $a_f$'s, 
we cook up an element $z \in S_N$ with the property that $\varphi_N(z) = 
\pi^{-N} v_0$.
Let $Z \subset S_N$ be the set of those elements.
From Proposition~\ref{prop:SGBtoGB} and Lemma~\ref{lem:syzygy}, we 
derive that $\bar S_N$ is generated by $\bar S_{N-1}$ (viewed as a 
submodule of $\bar S_N$ by filling new coordinates with zeroes)
and $Z$. Thus:
\begin{align*}
\varphi_N(S_N) 
 & = \varphi_{N-1}(S_{N-1}) + 
 \big< \varphi_N(Z), \, \pi{\cdot}\nu(Q_{\leq N}) \big> \\
 & \subset \varphi_{N-1}(S_{N-1}) + \big< \pi^{-N} Q_{> N}, \, 
 \pi{\cdot}\nu(Q_{\leq N}) \big>.
\end{align*}
The assertion~(b) now follows from the induction hypothesis, once
we have observed that $Q_{> N-1} = \pi^N \nu(Q_N) \cup Q_{> N}$.

Let us now prove~(c). Let $h \in I_{N+1}$. Then $h \in I_N$ and we
can use the induction hypothesis to write
$$h = \pi^N \hspace{-1ex} 
      \sum_{f \in Q_{\leq N}} \hspace{-1ex} a_f \nu(f) \hspace{1ex}
    + \sum_{g \in Q_{> N}} \hspace{-1ex} b_g g$$
for some $a_f, b_g \in \KzX$. Reducing modulo 
$\pi^{N+1}$, we find that the family $(a_f)_{f \in Q_{\leq N}}$ 
belongs to $S_N$. From~(b), we deduce that
$\sum_{f \in Q_{\leq N}} a_f \nu(f) \in
\big< \pi {\cdot} \nu(Q_{\leq N}), \, \pi^{-N} Q_{> N}\big>$.
We then conclude
by noticing that $Q_{>N} = \pi^{N+1} \nu(Q_{N+1}) \cup Q_{> N+1}$.

Finally, (d)~follows from~(c) by dividing by $\pi^{N+1}$ and
reducing modulo $\pi$.
\end{proof}

\noindent 
\textit{Termination.}
Since $\kX$ is noetherian, the sequence of ideals $(\bar I_N)$ is 
eventually constant. This implies that $\GB$ cannot grow indefinitely;
in other words, the final value of $\GB$ is reached in finite time.
However, the reader should be careful that this does not mean that 
Algorithm~\ref{algo:VaPoTe} terminates.
Indeed, once the final value of $\GB$ has been computed, one still 
has to check that the remaining series in $Q$ reduce to zero; this 
is achieved by performing divisions and can hang forever if we are 
working at infinite precision.
Nevertheless, this misfeature seems very difficult to avoid since, 
when working at infinite precision, the input series contain themselves 
an infinite number of coefficients and any modification on one of them 
could have a strong influence on the final result.

\medskip

\noindent
\textit{Correctness.}
Let $G$ be the output of Algorithm~\ref{algo:VaPoTe}, that is the
limit of the ultimately constant sequence $(\GB_N)$. For a positive
integer $N$, we define
$G_{\leq N}$ as the set of $f \in G$ with $\val(f) \leq N$.
Since only elements of valuation at least $N{+}1$ are added to 
$\GB$ after the checkpoint $\tau_{N+1}$, we deduce that $G_{\leq N} = 
\GB_{N+1}$.
Hence, by Proposition~\ref{prop:VaPoTe}, $\rho(G_{\leq N})$ is a GB 
of $\bar I_N$ for all $N \geq 0$. We are going to show that this
sole property implies that $G$ is indeed a GB of $I$.
For this, we consider $f \in I$. We write $N = \val(f)$,
so that $\rho(f)$ is the image in $\kX$ of $\pi^{-N} f$. Moreover,
we know that $\LM(\rho(f))$ is divisible by $\LM(\rho(g))$ for 
some $g \in G_{\leq N}$, \emph{i.e.} there exists $\i \in \NN^n$
such that $\LM(\rho(f)) = \X^\i {\cdot}\LM(\rho(g))$. This readily
implies that 
$\LM(f) = \pi^{N-\val(g)} \cdot \X^\i \cdot \LM(g)$, 
showing that $\LM(g)$ divides $\LM(f)$ in $\TTzX$ given that 
$\val(g) \leq N$. We have then proved that the leading monomial of 
any element of~$I$ is divisible by some $\LM(g)$ with $g \in G$,
\emph{i.e.} that $G$ is a GB of~$I$.

\section{Implementation}
\label{sec:implementation}

We have implemented both the PoTe and VaPoTe algorithms 
in \sage\footnote{\url{https://trac.sagemath.org/ticket/28777}}.
Our implementation includes the following optimization: at the end of the loop
(\emph{i.e.} after line~\ref{line:updateGB}), we minimize and reduce the current
GB in construction. This operation is allowed since all signatures are discarded
after each iteration of the loop.
Similarly, we reduce each new series $f$ popped from $Q$ on line~\ref{line:popf} 
before proceeding it.
These ideas were explored in the algorithm \hbox{F5-C}~\cite{EP10} and, as mentioned before, 
were one of the main motivations for adopting an incremental point of view.

Our implementation is also able to compute GB of ideals in $\KX$. 
For this, we simply use a reduction (for no extra cost) to the case 
of $\KzX$ (see~\cite[Proposition 2.23]{CVV}). We also normalize
the signatures in $S$ to be monic after each iteration of the main loop; 
in the PoTe algorithm,
this renormalization gives a stronger cover criterion and thus improves 
the performances.

As mentioned in Section~\ref{sec:vapote-algorithm}, Algorithm~\ref{algo:VaPoTe} remains
correct if the reductions are interrupted as soon as the valuation rises.
This can be done in the reduction step before processing the next $f$, before adding elements to the SGB, as well as in
the inter-reduction step.
Delaying reductions could be interesting, for instance,
if the input ideal is saturated: indeed, in this case, the algorithm never considers 
elements with positive valuation and delayed reductions do not need to be done
afterwards.
On the other hand, performing more reductions earlier leads to shorter reducers and
potentially faster reductions later.
In practice, in our current implementation, we have observed all possible scenarios:
interrupting the reductions can make the computation faster, slower,
or not make any significant difference.

\subsection{Some timings}

Numerous experimentations on various random inputs show that the VaPoTe algorithm
performs slightly better than the PoTe algorithm on average. Besides, both
PoTe and VaPoTe algorithms usually perform much better than Buchberger 
algorithm, although we observed important variations depending 
on the input system.

\begin{table}[t]
  \centering

\caption{Timings for the computation of GBs related to the torsion points on the Tate curve (all times in seconds)}
  \label{tab:benchs}

\vspace{-1mm}

  \begin{tabular}{l@{\hspace{1ex}}l@{\hspace{1ex}}l@{\hspace{5ex}}SSS}
    \multicolumn{3}{l}{Parameters} & {Buchberger} & {PoTe} & {VaPoTe} \\
    \midrule

    $p=5$, & $\ell=5$, & $\text{prec}=12$ &  87.9 & 72.2 & 19.2 \\
   $p=11$, & $\ell=5$, & $\text{prec}=12$ & 321   & 30.5 & 28.9 \\
$p=57637$, & $\ell=5$, & $\text{prec}=12$ &  83.2 & 13.3 & 13.3 \\
    $p=7$, & $\ell=7$, &  $\text{prec}=9$ &  62.3 & 45.3 & 27.7 \\
   $p=11$, & $\ell=7$, &  $\text{prec}=9$ & 168   & 36.0 & 28.5 \\
  \end{tabular}

\vspace{-1mm}

\end{table}

As mentioned in the introduction, Tate algebras are the building blocks
of $p$-adic geometry. 
One can then cook up interesting systems associated to meaningful 
geometrical situations. As a basic example, let us look at torsion
points on elliptic curves.

We recall briefly that (a certain class of) elliptic curves over 
$K = \Q_p$ are in one-to-one correspondence with a parameter $q$
lying in the open unit disc~\cite{TateCurve}. 
The parametric equation of these curves is
$y^2 + xy = x^3 + a_4(q)\: x + a_6(q)$ with:
$$a_4(q) = 5 \sum_{n=0}^\infty n^3 \frac{q^n}{1 - q^n},
\quad
a_6(q) = \sum_{n=0}^\infty \frac{7n^5 + 5n^3}{12} \frac{q^n}{1 - q^n}.$$
In order to fit with the framework of this article, we only consider 
parameters $q$ in the closed unit disc of radius $|p|$ and
perform the change of variables $q = p t$.
Given an auxiliary prime number $\ell$, we consider the $\ell$-th
division polynomial $\Phi_\ell(x,t) \in K\{t\}^\circ[x]$ associated to
the Weierstrass form of the above equation. By definition, its roots are 
the abscissas of $\ell$-torsion points of the Tate curve.
We now fix $p$ and $\ell$ and consider the system in $3$ variables
$\Phi_\ell(x, t_1) = \Phi_\ell(x, t_2) = 0$.
Its solutions parametrize the pairs of elliptic curves sharing a common 
$\ell$-torsion point. Computing a GB of it then provides information
about torsion points on $p$-adic elliptic curves.
Related (but more sophisticated) computations are likely to appear
in the study of the arithmetics of $p$-adic modular forms~\cite{Go88} 
or the development of $p$-adic analogues and refinements of Tate's isogeny 
Theorem~\cite{Ta66}.

Table~\ref{tab:benchs} shows the timings obtained for computing a GB of the above
systems for different values of $p$, $\ell$ and different precisions.
We clearly see on these examples that both PoTe and VaPoTe outperform 
the Buchberger algorithm.

\subsection{Towards further improvements}

\noindent
\textit{Faster reductions.}
Observing how our algorithms behave, one immediately notices that reductions
are very slow. It is not that surprising since our 
reduction algorithm is currently very naive. For this reason, we believe
that several structural improvements are quite possible.
An idea in this direction would be to store a well-chosen representative 
sample of reductions and reuse them later on. Typically, we could cache 
the reductions of all terms of the form $x_1^{2^{e_1}} \cdots 
x_n^{2^{e_n}}$ (with respect to the current GB in construction) and
use them to emulate a fast exponentation algorithm in the quotient
ring $\KzX/\left<GB\right>$.

Another attractive idea for accelerating reduction is to incorporate
Mora's reduction algorithm~\cite{Mora, RenEtc} in our framework.
Let us recall that Mora's algorithm is a special method for reducing
terms with respect to local or mixed orders (\emph{i.e.} orders for
which there exist terms $t < 1$), avoiding infinite loops in the
reduction process.
In our framework, infinite loops of reductions cannot arise since 
the computations are truncated at a given precision. Nevertheless,
we believe that Mora's algorithm can still be used to short-circuit 
some reductions.

The situation for Tate terms is actually significantly simpler than 
that of general local orders.
Indeed, Mora's reduction algorithm roughly amounts to add $\pi r$ to our list
of reductors each time we encounter a remainder $r$ (including $f$ itself)
in the reduction process. We believe that this optimization, if it is
carefully implemented, could already have some impact on the performances.
Besides, observing that the equality $f = r + \pi q f$ also reads 
$f = (1 - \pi q)^{-1} r$, we realize that Mora reductions of a Tate
series are somehow related to its Weierstrass decomposition. Moreover, 
at least in the univariate case, it is well known that Weierstrass 
decompositions can be efficiently computed using a well-suited Newton 
iteration. It could be interesting to figure out whether this strategy 
extends to multivariate series and, more generally, to the computation 
of arbitrary Mora reductions.

\medskip

\noindent
\textit{Using overconvergence properties.}
In a different direction, we would like to underline that the orderings
we are working with are by design block orders (comparing first the 
valuation). However, in the classical setting, we all know that graded 
orders often lead to much more efficient algorithms. 
Unfortunately, in the setting of this article, the very first definition 
of a Tate series already forces us to give the priority to the valuation 
in the comparison of terms; otherwise, the leading term would not be 
defined in general.

Nonetheless, we emphasize that even if graded orders do not exist over
$\KX$, they do exist over some subrings. Precisely, recall that, 
given a tuple $\r = (r_1, \ldots, r_n)$, we have defined\footnote{We
refer to \cite{CVV} for more details}:
$$\KX[\r] := \Big\{ \sum_{\i \in \NN^{n}} a_{\i}\X^{\i}
\text{ s.t. }
a_{\i}\in K \text{ and }
\val(a_\i) - \r{\cdot}\i\xrightarrow[|\i| \rightarrow +\infty]{} +\infty
\Big\}$$
where $\r{\cdot}\i$ denotes the scalar product of the vectors $\r$
and $\i$. When the $r_i$'s are all nonnegative, $\KX[\r]$ embeds
naturally into $\KX$; precisely, elements in $\KX[\r]$ are those
series that overconverges over the polydisk of polyradius 
$(|\pi|^{-r_1}, \ldots, |\pi|^{-r_n})$.
Moreover, the algebra $\KX[\r]$ is equipped with the valuation $\val_\r$ defined 
by:
$$\val_\r\Big(\sum_{\i \in \NN^{n}} a_{\i}\X^{\i}\Big) = 
\min_{\i \in \NN^{n}} \val(a_\i) - \r{\cdot}\i.$$
This valuation defines a new term ordering $\leq_{\r}$. 
We observe that, from the point of view of $\KX$, it really looks like a 
graded order: the quantity $\val_\r(f)$ plays the role of (the opposite 
of) a ``total degree'' which mixes the contribution of the valuation and 
that of the classical degree.

In light of the above remarks, we formulate the following question.
Suppose that we are given an ideal $I \subset \KzX$ (say, of
dimension $0$) generated by some series $f_1, \ldots, f_m$. If we have
the promise that the $f_i$'s all overconverge, \emph{i.e.} all lie in 
$\KX[\r]$ for a given $\r$, can we imagine an algorithm that computes a 
GB of $I$ taking advantage of the term ordering $\leq_{\r}$? 
As an extreme case, if we have the promise that all the $f_i$'s are
polynomials (that is $r_i = +\infty$ for all $i$), can one use this 
assumption to accelerate the computation of a GB of $I$?

\bibliographystyle{plain}

\end{document}


%% file: algos.tex

\DontPrintSemicolon

\begin{algorithm}[t]
  \Input{$f_1,\dots,f_m$ in $\kX$ (resp. $\KzX$)}
  \Output{a GB of the ideal generated by the $f_i$'s}

  $Q \leftarrow (f_{1},\dots,f_{m})$\;
  $\GB \leftarrow \emptyset$\;
  ~\label{line:blank}\;
  \For{$f \in Q$}{
    $G \leftarrow \{(0,g) : g \in \GB\} \cup \{(1,f)\}$\;
    $S \leftarrow \{\LM(g) : g \in \GB\}$\;
    $B \leftarrow \{ \Jpair((1,f), (0,g)) : g \in \GB \}$\;
    \While{$B \neq \emptyset$}{
      \Pop $(u,v)$ from $B$, with smallest $u$\;
      \lIf{$(u,v)$ is covered by $G$}{\Continue}
      \lIf{$u$ is divisible by some $s \in S$}{\Continue}
      $v_0 \leftarrow \regularreduce(u, v, G)$\;
      \If{$v_0 = 0$}{
        \Add $u$ to $S$\;
      } \Else{
        \For{$(s,g) \in G$}
        {\If{ $\Jpair((u,v_0), (s,g))$ is defined}
          {\Add $\Jpair((u,v_0), (s,g))$ to $B$\;}
          }
        \Add $(u,v_0)$ to $G$\;
      }
    }
    $\GB \leftarrow \{ v : (u,v) \in G\}$\label{line:updateGB}\;
  }
  \Return \GB\;

  \caption{G2V (resp. PoTe) algorithm}
\label{algo:PoTe}
\end{algorithm}

\begin{algorithm}[t]
  \Input{$f_1,\dots,f_m$ in $\KzX$}
  \Output{a GB of the ideal generated by the $f_i$'s}

  $Q \leftarrow (f_{1},\dots,f_{m})$\;
  $\GB \leftarrow \emptyset$\;
  \While{$Q \neq \emptyset$\label{line:while}}{
    \Pop $f$ from $Q$, with smallest valuation\label{line:popf}\;
    $G \leftarrow \{(0,g) : g \in \GB\} \cup \{(1,f)\}$\;
    $S \leftarrow \{\LM(g) : g \in \GB\}$\;
    $B \leftarrow \{ \Jpair((1,f), (0,g)) : g \in \GB \}$\;
    \While{$B \neq \emptyset$}{
      \Pop $(u,v)$ from $B$, with smallest $u$\;
      \lIf{\label{line:crit1}$(u,v)$ is covered by $G$}{\Continue}
      \lIf{\label{line:crit2}$u$ is divisible by some $s \in S$}{\Continue}
      $v_0 \leftarrow \regularreduce(u, v, G)$\label{line:regred}\;
      \If{$\val(v_0) > \val(f)$\label{line:diff1}}{
        \Add $u$ to $S$;\,
        \Add $v_0$ to $Q$\label{line:diff2}\;
      } \Else{
        \For{\label{line:addJpairs}$(s,g) \in G$}
        {\If{ $\Jpair((u,v_0), (s,g))$ is defined}
          {\Add $\Jpair((u,v_0), (s,g))$ to $B$\;}
          }
        \Add $(u,v_0)$ to $G$\;
      }
    }
    $\GB \leftarrow \{ v : (u,v) \in G\}$\;
  }
  \Return \GB\;

  \caption{VaPoTe algorithm}
\label{algo:VaPoTe}
\end{algorithm}
